\documentclass{article}
\usepackage{ijcai20}

\usepackage{times}

\usepackage{soul}
\usepackage{url}
\usepackage[hidelinks]{hyperref}
\usepackage[small]{caption}
\usepackage{graphicx}
\usepackage{amsmath}
\usepackage{booktabs}
\usepackage{algorithm}
\usepackage{algorithmic}
\usepackage{amsthm}
\usepackage{amsbsy, amsfonts, amsgen, amsmath, amsopn, amssymb, amstext,amsxtra, bezier, color, enumerate, graphicx, latexsym, verbatim,pictexwd, supertabular, url, dsfont,leftidx, mathrsfs, appendix,setspace}
\usepackage{bbm}
\newtheorem{lemma}{Lemma}
\newtheorem{theorem}{Theorem}

\title{Optimal Dispatch in Emergency Service System via Reinforcement Learning}

\author{
Cheng Hua\and
Tauhid Zaman
\affiliations
School of Management, Yale University\\
\emails
\{cheng.hua, tauhid.zaman\}@yale.edu
}

\begin{document}

\maketitle

\begin{abstract}
In the United States, medical responses by fire departments over the last four decades increased by 367\%. This had made it critical to decision makers in emergency response departments that existing resources are efficiently used. In this paper, we model the ambulance dispatch problem as an average-cost Markov decision process and present a policy iteration approach to find an optimal dispatch policy.  We then propose an alternative formulation using post-decision states that is shown to be mathematically equivalent to the original model, but with a much smaller state space. We present a temporal difference learning approach to the dispatch problem based on the post-decision states. In our numerical experiments, we show that our obtained temporal-difference policy outperforms the benchmark myopic policy.  Our findings suggest that emergency response departments can improve their performance  with minimal to no cost.
\end{abstract}

\section{Introduction}

In the United States, medical responses by fire departments over the last  four decades has increased by 367\% (Figure \ref{fig_EMS}), as reported by \cite{evarts2018fire}. The reasons for the dramatic increase in medical calls include the aging population and health insurance that covers most of the ambulance costs (Boston Globe, Nov 29, 2015). Cities with tight budgets are short of response units to respond to the growing amount of medical calls in time. NBC10 in Philadelphia, on Feb 28, 2019, reported that two thirds of the emergency medical calls had an ambulance response time of more than nine minutes. Thus, how to efficiently use the existing resources becomes an important topic to decision makers in emergency response departments. 

\begin{figure}[htbp]
\begin{center}
\includegraphics[width=8.5cm]{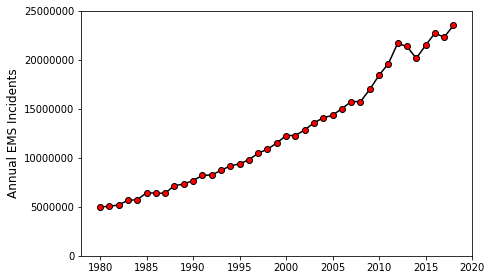}
\caption{Annual emergency medical servie (EMS) incidents from 1980 to 2018.}
\label{fig_EMS}
\end{center}
\vspace{-0.4cm}
\end{figure}

In current practice, cities dispatch ambulance units according to a fixed dispatch rule, which always dispatches the closest available unit. The dispatch policy is deemed as a myopic policy as it only considers the current call and ignores the impact of dispatching a unit to future calls. In this paper, we model the ambulance dispatch problem as an average-cost Markov decision process (MDP).  We finds a dispatch policy that minimizes the mean response time using reinforcement learning. We propose an alternative MDP formulation for the problem using post-decision states that we show is mathematically equivalent to the original MDP formulation, but with a much smaller state space. 

Due to the curse of dimensionality, the application of both formulations are only restricted to small problems. To solve larger problems, we use temporal difference learning (TD-learning)  with the post-decision states. In our numerical experiments, we show that the TD-learning algorithm converges quickly and the policies obtained from our method outperform the myopic policy that always dispatches the closest available unit.

The remainder of this paper is organized as follows. In Section \ref{sec_literature_review}, we provide a review of the relevant literature. In Section \ref{sec_MDP}, we present the Markov decision process formulation. Section \ref{sec_post_decision} presents the formulation using post-decision states which reduces the state space of the original formulation. In Section \ref{sec_TD_learning}, we present the temporal difference learning algorithm that is based on the post-decision states and its theoretical properties, while in Section \ref{sec_numerical}, we show the performance of the algorithm in numerical experiments. We conclude the paper in Section \ref{sec_conclusion}.

\section{Literature Review}
\label{sec_literature_review}
There are a number of papers that study dispatch policies for emergency medical services (EMS) in different settings. The optimal dispatch rule  was first studied in \cite{carter1972response}. The authors studied a simple system of two units and provided a closed form solution that determines the response areas to be served by each unit to minimize average response time. However, such a closed form solution no longer exists in a system that has more than two units, and finding the optimal dispatch rule has been an important topic. 

\cite{bandara2014priority} developed a simulation model to evaluate priority dispatching strategies for ambulance systems with different priorities of calls to improve patient survival probability. The optimal policy was found via full enumeration using a simulation model, and the authors found that dispatching the closest available unit is not always optimal in this setting. The limitation of their method is due to the exponential computational complexity. A system with more than four units would take a significant amount of time. \cite{mclay2013model} developed a Markov decision process to optimally dispatch ambulances to emergency calls with classification errors in identifying priorities to maximize the expected coverage of true high-risk patients. They compared the optimal policy to the myopic policy that sends the closest unit and found an advantage in increasing the utility. The authors applied the value iteration algorithm to find the optimal solution, but the solution method is also limited only to small problems. They extended the paper in \cite{mclay2013dispatching} to consider the optimal policy that balances equity and efficiency by adding four equity constraints. 

\cite{jagtenberg2017dynamic} studied whether dispatching the closest available unit is optimal in a dynamic ambulance dispatching problem based on a Markov decision process. The problem was discretized by time using one minute as the time interval. Two objectives were considered: mean response time and the fraction of calls with response time beyond a certain time threshold. The value iteration method was used to find the optimal solution. A heuristic that leverages the dynamic maximum expected covering location problem (MEXCLP) model was proposed by \cite{jagtenberg2015efficient}, where the MEXCLP problem was studied by \cite{daskin1983maximum} to determine the best unit locations using expected covering demand as a proximal objective function than response time. They found an improvement of 18\%  on the fraction of late arrivals using the optimal policy compared to the myopic policy, which provides insights about deviating from the myopic policy. \cite{jagtenberg2017benchmarking} empirically provide a bound for the gap between the existing solutions and the optimal solution.  

The aforementioned papers use enumeration or MDP models to analyze dispatch policies. While they are guaranteed to find the optimal policy, their computational complexity limit their use to only small systems, due to the curse of dimensionality. For larger problems, they normally fail badly computationally. Thus, many works have tried to resolve this issue using approximate dynamic programming (ADP). A comprehensive introduction to the application of Approximate Dynamic Programming in the operations research filed can be found in \cite{powell2010merging} and \cite{powell2007approximate}.

\cite{schmid2012solving} followed the same formulation as introduced in \cite{powell2010merging} that uses ADP with aggregation function approximation to an ambulance relocation and dispatching problem to reduce the mean response time. A seminal paper by \cite{maxwell2010approximate} applied ADP to the ambulance redeployment problem, where they used Monte Carlo simulation with one-step bootstrap to estimate complex expectations and applied least squares regression with linear function approximation to learn the approximate value functions. \cite{nasrollahzadeh2018real} studied the problem of real-time ambulance dispatching and relocation, which is formulated as a stochastic dynamic program and solved using ADP. Their formulation and method is the same as proposed in  \cite{maxwell2010approximate}. Issues of this approach include that while most of the time they can beat the benchmark policy, they normally never output the optimal policy. Also, it is not guaranteed that the learning method always converges and finding useful basis functions for approximation is more of an art than science, which requires domain knowledge and testing. 

Our paper is different from the existing work in the following three aspects: 1) We model the problem as an average-cost MDP problem, which solves for the optimal policy that minimizes the mean response time.  The papers we mentioned above use simpler discounted-reward schemes which minimize the sum of discounted response times as a surrogate to the mean response time.  However, discounted response times are more difficult to interpret compared to mean response time. 2) We propose an alternative formulation with post-decision states and show its mathematical equivalence to the original problem formulation. 3) We show that the proposed TD-learning algorithm based on post-decision states is guaranteed to converge to the optimal solution, while little or no theoretical guarantees exist in the aforementioned works using ADP methods. 

\section{Markov Decision Process Formulation}
\label{sec_MDP}
Consider a geographical region $R\subset \mathbb{R}^2$ that is served by a total of $N$ emergency units, all of the same type, e.g., ambulance units. Calls arrive in region $R$ according to a Poisson point process with an arrival intensity $\{\Lambda(x, y): (x, y) \in R\}$ at location with coordinate $(x,y)$. We partition the region $R$ into $J$ sub-regions and associate a center of mass with each sub-region $R_j \subset R$, which is also referred to as demand \textit{node} $j$. Note that $J$ can be as large as needed. Denote the total call rate in node $j$ as 
\begin{equation}
    \lambda_j=\int_{R_j} \Lambda(x, y) dxdy.
\end{equation}
The overall call rate in region $R$ is denoted by 
\begin{equation}
    \lambda = \sum_j \lambda_j = \int_R \Lambda(x, y) dxdy.
\end{equation}
We assume the mean service time follows a negative exponential distribution with rate $\mu_i$ for unit $i$. We assume that the service time includes turnout time, travel time and time at the scene. The justification for this assumption is that travel time is usually a small portion of the total service time. With longer travel times, \cite{jarvis1975optimization} mentioned a \textit{mean service time calibration} to calibrate the service time to maintain the Markov property. 

Let $b_i$ represent the state of unit $i$, where $b_i = 0$ if the unit is available and $b_i = 1$ if the unit is busy. We denote the state space of all units as an $N$-dimensional vector $B=\{b_{N}\cdots b_{1}\} \in \mathscr{B}$, which is in a backward order similar to the representation of binary numbers. We define $B=b_i$ as the status of unit $i$. Note that $|\mathscr{B}|=2^N$. 

If all units are busy when a call is received, we assume that it is handled by a unit outside of the system, such as a private ambulance company, or a unit from a neighboring jurisdiction, which is common mutual aid policy. 

Define $t_{ij}$ as the average response time from the base of unit $i$ to node $j$. In this paper, we aim to find the optimal dispatch policy that minimizes the average response time of all served calls. 

\subsection{State Space}
Define $S$ as the state space $S$ and $s \in S$ as the state of the system. We have 
\begin{equation}
s = (j, B),
\end{equation}
which is a tuple of $j$ and $B$ that consists of the location of the current call and the state of all units (available or busy) at the time of the arrival. We denote $s(0)=j$ and $s(1)=B$ in state $s$. The entire state space has size $|S| = J\times 2^N$.


\subsection{Action Space}
When a call is received in the system, we decide on which unit to be dispatched to this call. An action in this problem is to dispatch a particular unit upon receiving a call, so the action space is given as 
\begin{equation}
A = \{1,2,\cdots, N\}.
\end{equation}
Note that only an available unit may be dispatched. We define $A_s \subset A$ as the set of feasible actions at state $s$, where  
\begin{equation}
A_s = \{i: B(i)=0, i=\{1,2,\cdots, N\}\}.
\end{equation}
We define $a \in A_s$ as an action from the feasible action space. 

\subsection{Policy Space}
We define the policy space as $\Pi$, the set of all feasible policies. A policy $\pi \in \Pi$ is a mapping from the state space to the action space, $\pi: S \rightarrow A$. Specifically, $\pi(s)=a, a\in A_s$. The optimal policy $\pi^*\in \Pi$ is the policy that minimizes the average cost over all time steps. Our goal is to find this optimal policy. 

We use a benchmark policy which sends the closest available unit, denoted by $\pi^m \in \Pi$. 
\begin{equation} 
\pi^m(s) = \arg\min_i t_{ij},  \qquad  \forall i\in A_s,  \forall s\in S. 
\end{equation}

Sending the closest available unit is a policy widely used in emergency service systems. This policy is myopic as it does not consider potential future calls. Saving the closest unit to the current call might greatly reduce the expected response time of a future call. 

\subsection{Costs}
Define $c^{\pi}(s)$ as the cost of being in state $s$ following policy $\pi$, which equals to the response time
\begin{equation}
    c^{\pi}(s) = t_{ij}
\end{equation}
when the call location in state $s$ is $j$, $s(0)=j$, and the policy dispatches unit $i$ in state $s$, $\pi(s)=i$.

\subsection{Transition Probabilities with Augmented Transitions}
\label{subsec_trans_p}
Define $p^{\pi}(s,s')$ as the transition probability from state $s=(j,B)$ to state $s'=(j',B')$ under policy $\pi$. In determining the transition state probability, we consider an augmented transition where a unit completes a service and no dispatch action is needed. This is because the number of service completed between two arrivals is a random variable whose probability is complicated to compute. Introducing the augmented transition reduces the number of transition possibilities. Denote $I_i$ as the vector of all 0's except for the $i$th element, which is 1. The transition rate $p^{\pi}(s,s')$ with augmented transition is given as
\begin{equation}
p^{\pi}(s,s') = 
\begin{cases}
     \frac{\lambda_{j'}}{\lambda+\sum_{k:B(k)=1}\mu_k+\mu_i}, & \text{if } s' = (j', B+I_i),\\
     \frac{\mu_l}{\lambda+\sum_{k:B(k)=1}\mu_k+\mu_i}, & \text{if } s' = (\emptyset, B+I_i-I_l).
    \end{cases} 
\end{equation}
where the expression on the top corresponds to the transition from state $s$ to state $s'$ upon action $\pi(s)=i$ is taken and a new call arrives in node $j'$, and the bottom expression corresponds to the augmented transition where no arrival occurs but a busy unit $l\in A/A_s$ completes its current service. No action is needed since there is no arriving calls in this transition. 

\subsection{Bellman's Equation}

Define $V^{\pi}:S^+ \mapsto \mathbb{R}$ as the value function for the MDP following policy $\pi$ and the value of state $s$ is $V^{\pi}(s)$, where $S^{+}$ is the augmented state space that has dimension $|S^{+}| = (J+1)2^N$. Let $\mu^{\pi}$ be the average cost following policy $\pi$. The Bellman's equation for the average cost is
\begin{equation}
V^{\pi}(s) =  c^{\pi}(s) - \frac{\mu^{\pi}}{2} + \sum_{s'\in S} p^{\pi}(s,s') V^{\pi}(s'), \ \forall s \in S^{+}. \label{eq_bellman_perstate}
\end{equation}
Note that the $1/2$ in the above equation is due to the existence of augmented transitions. A transition that is due to a service completion has zero cost and the number of service completions is always equal to the number of calls being served. 

Define $V^{\pi}$ as the vector of all state values, $c^{\pi}$ as the vector of all state costs, $P^{\pi}$ as the transition matrix, and $e$ as the vector of all ones. The vector form of Bellman's equation is 
\begin{equation}
    V^{\pi} = c^{\pi} - \frac{\mu^{\pi} e}{2} + P^{\pi} V^{\pi}. \label{eq_bellman}
\end{equation}
The solution to the above Bellman's equation is not unique. Instead, the set of all value functions takes the form $\left\{V^{\pi}+m e \mid m \in \mathbb{R}\right\}$. Since shifting the value function by a constant value does not change the relative differences between state values, once we obtain a set of state values $V^{\pi}$, the policy can be updated as 
\begin{equation}
    \pi'(s) = \arg \min_{a\in A_s} t_{aj}+\sum_{s'\in S^{+}} p(s,s'|a) V^{\pi}(s') \label{eq_policy_update},
\end{equation}
where $p^{\pi}(s,s'|a)$ is the one-step transition probability when taking action $a$ instead of following the policy $\pi(s)$. This so-called policy iteration method for solving the MDP is summarized in Algorithm \ref{algo_MDP}.

\begin{algorithm}
\caption{Policy Iteration Method}
\label{algo_MDP}
\begin{algorithmic}[1]
\STATE Pick a random policy $\pi_0$. Set $k=0$.
\WHILE{$\pi_k \not = \pi_{k+1}$}
\STATE Compute the cost $c^{\pi_k}$ and transition matrix $P^{\pi_k}$.
\STATE \textbf{Policy Evaluation:} Solve the state values $V^{\pi_k}$ from the Bellman's equation \eqref{eq_bellman}.
\STATE \textbf{Policy Improvement:} For each state $s\in S$, update the actions of each state by 
\begin{equation*}
    \pi_{k+1}(s) = \arg \min_{a\in A_s} t_{aj}+\sum_{s'\in S^{+}} p(s,s'|a) V^{\pi_k}(s').
\end{equation*}
\STATE $k = k+1$
\ENDWHILE
\STATE \textbf{Output:}  Optimal Policy $\pi^* = \pi_k$
\end{algorithmic}
\end{algorithm}

\section{Post-Decision State Formulation}
\label{sec_post_decision}
The policy iteration method guarantees the convergence to the optimal dispatch policy that minimizes the average response time, which requires solving a linear system with $(J+1)2^N$ states repeatedly. In the section, by realizing the nature of the state transitions of the dispatch problem, we introduce the notion of post-decision states and use them as the new states in our problem. We show that the MDP formulation using post-decision states reduce the state space to $2^N$, which also guarantees to find the optimal dispatch policy. In the next section, we will develop the temporal difference learning method based on this formulation using post-decision states. 

In the original formulation, a state is a tuple of call locations and the statuses of all units $s=(j, B)$. A post-decision state $s_x$ is a state that the system is in immediately after observing state $s$ and taking action $a = \pi(s)$, before the next random information arrives into the system, which is the arrival of the next call location. Thus, given state $s$ and unit $i$ being dispatched, i.e., $a=\pi(s)=i$, the post-decision state $s_x$ is 
\begin{equation}
    s_x = B + I_i.
\end{equation}
Note that by defining the post-decision state this way, we only need  information about the statuses of all units. Define $S_x$ as the post-decision state space; we have $|S_x|=2^N$. Indeed, $S_x = \mathscr{B}$. 

\begin{lemma}
Let $p_x^{\pi}(s_x,s_x')$ be the corresponding transition probability from $s_x$ to $s_x'$. We have
\begin{equation}
p_x^{\pi}(s_x,s_x') = \begin{cases}
\frac{\sum_{j\in \mathscr{R}_{l|s_x}^{\pi}} \lambda_j}{\lambda+\sum_{k:B(k)=1}\mu_k+\mu_i},  & \text{if } s_x' = s_x+I_l,\\
\frac{\mu_l}{\lambda+\sum_{k:B(k)=1}\mu_k+\mu_i}, & \text{if } s_x' = s_x-I_l,
\end{cases}
\end{equation}
where $\mathscr{R}_{l|s_x}^{\pi}$ is the set of demand nodes where policy $\pi$ dispatches unit $l$, i.e.,
\begin{equation}
    \mathscr{R}_{l|s_x}^{\pi} = \{j: \pi(s')=l, s'=(j,s_x)\}.
\end{equation}
\end{lemma}

\begin{proof}
For $s_x' = B+I_i-I_l$, where a transition happens when unit $l$ completes its current service, the post-decision state transition is the same as the transition from $s$ to the augmented state $s'= (\emptyset, B+I_i-I_l)$ as no call arrives and no action is needed for this state; thus $s'=s_x'$. 

For $s_x' = B+I_i+I_l$, where a call arrives in post-decision state $s_x$, we need to capture the randomness of exogenous information, which is the location of call that arrives in $s_x$. We thus have 
\begin{eqnarray*}
    p_x^{\pi}(s_x,s_x') &=& \sum_j p^{\pi}(s,s'= (j, s_x)) \mathbbm{1}_{\{\pi(s')=l\}}\\
    &=& \sum_j p^{\pi}(s,s') \mathbbm{1}_{\{s'= (j, s_x)\}} \mathbbm{1}_{\{\pi(s')=l\}}\\
    &=& \sum_{j\in \mathscr{R}_{l|s_x}^{\pi}} p^{\pi}(s,s'= (j, s_x))\\
    &=& \frac{\sum_{j\in \mathscr{R}_{l|s_x}^{\pi}} \lambda_j}{\lambda+\sum_{k:B(k)=1}\mu_k+\mu_i}
\end{eqnarray*}
\end{proof}

\begin{lemma}
The cost of post-decision state $s_x$, $c^{\pi}(s_x)$, is given as
\begin{equation}
    c^{\pi}(s_x) = \frac{\sum_{l\in A_s} \sum_{j\in \mathscr{R}_{l|s_x}^{\pi}}\lambda_j t_{lj}}{\lambda+\sum_{k:B(k)=1}\mu_k+\mu_i}.
\end{equation}
\end{lemma}
\begin{proof}
The cost of $s_x$ is the expected one-step transition cost from $s_x$ to $s_x'$ under policy $\pi$. Let $c^{\pi}_l(s_x)$ be the expected cost of dispatching unit $l$ in $s_x$. We have
\begin{equation}
    c^{\pi}_l(s_x) = \begin{cases}
    \frac{\sum_{j\in \mathscr{R}_{l|s_x}^{\pi}}\lambda_j t_{lj}}{\sum_{j\in \mathscr{R}_{l|s_x}^{\pi}}\lambda_j}, & \text{if } l\in A_s,\\
     0, & \text{if } l\in \emptyset.
    \end{cases} 
\end{equation}
We thus have
\begin{eqnarray*}
    c^{\pi}(s_x) &=& \sum_{l\in A_s}c^{\pi}_l(s_x)p_x^{\pi}(s_x,s_x+I_l)\\
    &=&  \sum_{l\in A_s}  \frac{\sum_{j\in \mathscr{R}_{l|s_x}^{\pi}}\lambda_j t_{lj}}{\sum_{j\in \mathscr{R}_{l|s_x}^{\pi}} \lambda_j} \frac{\sum_{j\in \mathscr{R}_{l|s_x}^{\pi}} \lambda_j}{\lambda+\sum_{k:B(k)=1}\mu_k+\mu_i}\\
    &=&\frac{\sum_{l\in A_s} \sum_{j\in \mathscr{R}_{l|s_x}^{\pi}}\lambda_j t_{lj}}{\lambda+\sum_{k:B(k)=1}\mu_k+\mu_i}.
\end{eqnarray*}
\end{proof}

Note that all components defining $c^{\pi}(s_x)$ and $p_x^{\pi}(s_x,s_x')$ are known, which are computed beforehand. Define $J^{\pi}:S_x \mapsto \mathbb{R}$ as the value function for the post-decision state space $S_x$. Let $\mu_x^{\pi}$ be the average cost following policy $\pi$ in the post-decision state space. The Bellman's equation is
\begin{equation}
J^{\pi}(s_x) =  c^{\pi}(s_x) - \frac{\mu_x^{\pi}}{2} + \sum_{s_x'\in S_x} p_x^{\pi}(s_x,s_x') J^{\pi}(s_x'), \ \forall s_x \in S_x.
\end{equation}
Let $J^{\pi}$, $c_x^{\pi}$ and $P_x^{\pi}$ be the corresponding vector representations. The vector form Bellman's equation around post-decision states is 
\begin{equation}
    J^{\pi} = c_x^{\pi} - \frac{\mu_x^{\pi} e}{2} + P_x^{\pi} J^{\pi}. \label{eq_bellman_post}
\end{equation}

\begin{theorem}
    The MDP formulation around post-decision states is equivalent to the original formulation. In particular
    \begin{enumerate}[i)]
        \item $\mu_x^{\pi} =  \mu^{\pi}$;
        \item For $s_x=B$, let $\Gamma = \lambda+\sum_{k:B(k)=1}\mu_k$. We have
        \begin{eqnarray}
            J^{\pi}(s_x) =\sum_{j}\frac{\lambda_j }{\Gamma}V^{\pi}\big(s^{(j)}\big)+\sum_{k:B(k)=1}\frac{\mu_k }{\Gamma}V^{\pi}\big(s^{[k]}\big), \label{eq_value_quivalence}
        \end{eqnarray}
        where $s^{(j)}=(j,B)$ and $s^{[k]}=(\emptyset,B-I_k)$.
    \end{enumerate}
\end{theorem}
\begin{proof}
Expanding the value functions $V^{\pi}$ in \eqref{eq_value_quivalence} by \eqref{eq_bellman_perstate} and collecting terms, we have
\begin{eqnarray*}
    && \sum_{j}\frac{\lambda_j}{\Gamma}V^{\pi}\big(s^{(j)}\big) +\sum_{k:B(k)=1}\frac{\mu_k }{\Gamma}V^{\pi}\big(s^{[k]}\big)\\
    &=& \sum_{j}\frac{\lambda_j}{\Gamma}c^{\pi}\big(s^{(j)}\big)  - \sum_{j}\frac{ \lambda_j}{\Gamma}\frac{\mu^{\pi}}{2} \\
    && + \sum_{j}\frac{ \lambda_j}{\Gamma} \sum_{s'\in S} p^{\pi}(s^{(j)},s') V^{\pi}(s')\\
    && + \sum_{k:B(k)=1}\frac{\mu_k }{\Gamma} \underbrace{c^{\pi}\big(s^{[k]}\big)}_{=0} - \sum_{k:B(k)=1}\frac{\mu_k}{\Gamma} \frac{\mu^{\pi}}{2} \\
    && + \sum_{k:B(k)=1}\frac{\mu_k}{\Gamma} \sum_{s'\in S} p^{\pi}(s^{[k]},s') V^{\pi}(s')\\
    &=& \underbrace{\frac{\sum_{l\in A_{s^{(j)}}} \sum_{j\in \mathscr{R}_{l|s_x}^{\pi}}\lambda_j t_{lj}}{\Gamma}}_{c^{\pi}(s_x)} - \Big(\underbrace{\sum_{j}\frac{\lambda_j}{\Gamma} + \sum_{k:B(k)=1}\frac{\mu_k}{\Gamma}}_{=1}\Big) \frac{\mu^{\pi}}{2} \\
    && + \sum_{j}\frac{ \lambda_j}{\Gamma} \Big[\frac{\sum_{j'} \lambda_{j'}}{\Gamma+\mu_{\pi(s^{(j)})}} V^{\pi}\big(s^{(j')}\big)+  \frac{\sum_{l:B(l)=1} \mu_l}{\Gamma+\mu_{\pi(s^{(j)})}} V^{\pi}\big(s^{[l]}\big)\Big]\\
    && + \sum_{k:B(k)=1}\frac{\mu_k}{\Gamma} \Big[\frac{\sum_{j'} \lambda_{j'}}{\Gamma} V^{\pi}\big(s^{(j')[k]}\big)\\
    && \qquad\qquad\qquad\qquad +\frac{\mu_k}{\Gamma} \frac{\sum_{l:B(l)=1} \mu_k}{\Gamma} V^{\pi}\big(s^{[k][l]}\big)\Big]\\
    &=& c^{\pi}(s_x) - \frac{\mu^{\pi}}{2} + \sum_{s_x'\in S_x} p_x^{\pi}(s_x,s_x') J^{\pi}(s_x')
\end{eqnarray*}
The last equality holds by realizing the first part of the summation corresponds to the transition from a call arrival and dispatching a unit, and the the last part corresponds to the transition from a service completion that a unit returns to its base and becomes available. 
\end{proof}

Under this formulation, the new policy $\pi'$ for state $s=(j, B)$ is updated as
\begin{equation}
    \pi'(s) = \arg \min_{a\in A_s} t_{aj}+ J^{\pi}(s_x= B + I_a).
\end{equation}
We note that even though the state space is smaller using post-decision states, the policy state space remains the same. The new policy iteration around post-decision states is summarized in Algorithm \ref{algo_MDP_revised}.

\begin{algorithm}
\caption{Policy Iteration with Post-Decision States}
\label{algo_MDP_revised}
\begin{algorithmic}[1]
\STATE Pick a random policy $\pi_0$. Set $k=0$.
\WHILE{$\pi_k \not = \pi_{k+1}$}
\STATE Compute the cost $c_x^{\pi_k}$ and transition matrix $P_x^{\pi_k}$.
\STATE \textbf{Policy Evaluation:} Solve the state values $J^{\pi_k}$ from the Bellman's equation \eqref{eq_bellman_post}.
\STATE \textbf{Policy Improvement:} For each state $s\in S$, update the actions of each state by 
\begin{equation*}
    \pi_{k+1}(s) = \arg \min_{a\in A_s} t_{aj}+ J^{\pi_k}(s_x= B + I_a).
\end{equation*}
\STATE $k = k+1$
\ENDWHILE
\STATE \textbf{Output:}  Optimal Policy $\pi^* = \pi_k$
\end{algorithmic}
\end{algorithm}

\section{Temporal Difference Learning with Post-Decision States}
\label{sec_TD_learning}
The policy iteration method with post-decision states is also guaranteed to find the optimal policy while reducing the number of equations from $(J+1)2^N$ to $2^N$. However, the exponential complexity limits its application only to small or medium-sized problems. In this section, we present the temporal difference learning method using value function approximation based on the formulation with post-decision states that reduces the dimension of the problem. 

Let $\phi_{[p]}: S_x \mapsto \mathbb{R}$, $p=1,2\cdots,P$, be the basis functions of post-decision states, and let $r = \{r_{[p]}: p=1,2\cdots,P\}$ be the tunable parameters. The value function approximation is given by $\tilde{J}(s_x, r)$, where
\begin{equation}
    \tilde{J}(s_x, r)=\sum_{p=1}^{P} r_{[p]} \phi_{[p]}(s_x).
\end{equation}

Let $\tilde{J}(r)$ be the vector of approximate state values of all states given parameter vector $r$ and let $\Phi$ be an $2^N \times P$ matrix whose pth column is equal to the vector $\phi_{[p]}$ of all states in $S^x$. The vector form of the above equation is 
\begin{equation}
    \tilde{J}(r) = \Phi r.
\end{equation}

Define $\left\{x_{t} \mid t=0,1, \ldots\right\}$ as the Markov chain on the post-decision state space $S_x$ with transition matrix $P_x^{\pi}$.
\begin{lemma}
\label{lemma_MC}
The Markov chain corresponding to $P_x^{\pi}$ is
irreducible and has a unique stationary distribution.
\end{lemma}

The proof of the above lemma is straightforward by noting that the Markov chain with post-decision state space forms a hypercube loss model whose property can be found in \cite{larson74hypercube}. We define the temporal difference $d_t$ as
\begin{equation}
    d_{t}=c\left(x_{t}\right)-\frac{\mu_{t}}{2}+\tilde{J}\left(x_{t+1}, r_{t}\right)-\tilde{J}\left(x_{t}, r_{t}\right),
\end{equation}
where $c\left(x_{t}\right)-\frac{\mu_{t}}{2}+\tilde{J}\left(x_{t+1}, r_{t}\right)$ is the differential cost function at state $x_t$ based on the one-step bootstrap and $\tilde{J}\left(x_{t}, r_{t}\right)$ is the old approximate differential cost function at state $x_t$. Define $r_{t}$ as the parameter vector at time $t$. We update the  parameter vector $r$ using this temporal difference by
\begin{equation}
    r_{t+1}=r_{t}+\gamma_{t} d_{t} \phi\left(x_{t}\right),
\end{equation}
\begin{equation}
    \mu_{t+1}=\left(1-\gamma_{t}\right) \mu_{t}+2\gamma_{t} c\left(x_{t}\right).
\end{equation}
We let $\gamma_{t} = \frac{a}{a+t}$ where $a\geq 1$ is a hyper-parameter that controls the learning speed. We note that when $a=1$, $\gamma_{t}$ forms the harmonic series, but usually the performance is not good empirically as discussed in \cite{powell2007approximate}. A larger $a$ slows the learning speed at the beginning of the process and choosing the appropriate value of $a$ depends on the nature of the problem. 

In this paper, we present a simple way of defining the basis function. We give an index $i_{x}$ to each post-decision state $s_x=B$, where $i_x = \sum_i^N 2^{i} \mathbbm{1}_{B(i)=1}$. We let the basis function $\phi_{[p]}(s_x)$ be 
\begin{equation}
    \phi_{[p]}(s_x) = \begin{cases}
     1, & \text{if } p = i_x,\\
     0, & \text{if } p\not= i_x.
    \end{cases} 
\end{equation}

The algorithm is described in Algorithm \ref{algo_TD_Learning}. 

\begin{algorithm}[htbp]
\caption{Policy Iteration with TD-Learning}
\label{algo_TD_Learning}
\begin{algorithmic}[1]
\STATE Pick a random policy $\pi_0$. Set $k=0$. 
\STATE Specify $T$ and $K$.
\WHILE{$k\leq K$}
\STATE Set $t=0$. Initialize $r_0$ and $\mu_0$. 
\STATE Starting from a random state $x_0$, generate a state trajectory $\left\{x_{t} \mid t=0,1, \ldots T\right\}$ corresponding to the Markov chain with state transition probability $P_x^{\pi_k}$ that is defined by the policy $\pi_k$. 
\FOR{$t = 0$ to $T$}
    \STATE Calculate the temporal difference $d_t$ by
    \begin{equation*}
        d_{t}=c\left(x_{t}\right)-\frac{\mu_{t}}{2}+\tilde{J}\left(x_{t+1}, r_{t}\right)-\tilde{J}\left(x_{t}, r_{t}\right).
    \end{equation*}
    \STATE Update the parameters by
    \begin{eqnarray*}
        r_{t+1}&=&r_{t}+\gamma_{t} d_{t} \phi\left(x_{t}\right),\\
        \mu_{t+1}&=&\left(1-\gamma_{t}\right) \mu_{t}+2\gamma_{t} c\left(x_{t}\right).
    \end{eqnarray*}
\ENDFOR
\STATE For each state $s\in S$, update the policy
\begin{equation*}
    \pi_{k+1}(s) = \arg \min_{a\in A_s} t_{aj}+ \tilde{J}(s_x=B+I_a, r_T).
\end{equation*}
\STATE $k = k+1$
\ENDWHILE
\STATE \textbf{Output:} Policy $\tilde{\pi}^* = \pi_K$
\end{algorithmic}
\end{algorithm}

\begin{theorem}
\label{thm_convergence}
Algorithm \ref{algo_TD_Learning} has the following property:
\begin{enumerate}
    \item Converges with probability 1.
    \item The limit of the sequence $\frac{\mu_t}{2}$ at the $k$th iteration of the algorithm is the average cost $\frac{\mu_x^{\pi_k}}{2}$, i.e.,
    \begin{equation}
        \lim_{t\rightarrow \infty} \mu_t = \mu_x^{\pi_k}.
    \end{equation}
    \item The limit of the sequence $r_t$ at the $k$th iteration of the algorithm, denoted by $r^{k*}$, is the unique solution of the equation
    \begin{equation}
        T\left(\Phi r^{k*}\right)=\Phi r^{k*},
    \end{equation}
    where $T: \mathbb{R}^{2^N} \mapsto \mathbb{R}^{2^N}$ is an operator defined by 
    \begin{equation}
        T J = c_x^{\pi_k}-\frac{\mu^{\pi_k}_x e}{2}+P_x^{\pi_k} J.
    \end{equation}
\end{enumerate}
\end{theorem}
The proof of the above theorem follows from Lemma \ref{lemma_MC}, and the basis functions $\phi(s_x)$ being linearly independent for all states. Also it is necessary to have the sequence $\gamma_t$ is positive, deterministic, and satisfies $\sum_{t=0}^{\infty} \gamma_{t}=\infty$ and $\sum_{t=0}^{\infty} \gamma_{t}^{2}<\infty$. A detailed proof are shown in \cite{tsitsiklis1999average}. 

Theorem \ref{thm_convergence} guarantees that the TD-Learning algorithm with post-decision states shown in Algorithm \ref{algo_TD_Learning} always returns the optimal policy if $T$ is large enough. When $T$ is moderately large, it is enough to obtain a policy close to optimal as we will show in the next section. Our algorithm avoids solving the complex bellman's equation which has an exponential complexity. Once calculated the cost vector $c_x^{\pi}$ and transition matrix $P_x^{\pi}$, we easily obtain the temporal differences $d_t$ by Monte Carlo simulation and evaluate the value functions that is needed for policy improvement in the policy iteration algorithm. One can also leverage parallel computing to obtain multiple sample paths at the same time, which further reduces the computation time. 

\section{Numerical Results}
\label{sec_numerical}
In this section, we show the numerical results comparing the policy obtained from the TD-Learning method to the myopic policy that always dispatch the closest available unit for systems with $N=5, 10$ and $15$ units.  We created an imaginary region which is partitioned into $J=30$ demand nodes. We randomly locate units in the region and obtain the corresponding response times from each unit to each demand point. 

The myopic policy $\pi^m$ is obtained by choosing from the available units in each state $s\in S$ that results in the shortest response time. We solve the exact mean response time of the myopic policy using the hypercube queuing model developed by \cite{larson74hypercube} for the cases of $N=5$ and $10$ units. For the case where $N=15$, the hypercube model becomes computationally costly due to its exponential complicity, and we obtain the approximate mean response time using the approximation method developed by \cite{larson75approx}, which has been tested to be very close to the exact value. An alternative is to obtain the approximate result via simulation which could be more accurate while taking a longer time. 

The policy from the proposed TD-Learning method with post-decision states is obtained by running the algorithm in 25 iterations. We perform a roll-out with 200,000 state transitions in each iteration and update the parameter vector $r$ using the temporal differences $d$. We record the sample average response time in each iteration and the results are shown in Figures \ref{fig_N_5}, \ref{fig_N_10}, and \ref{fig_N_15}, respectively. 

\begin{figure}[htbp]
\begin{center}
\includegraphics[width=7cm]{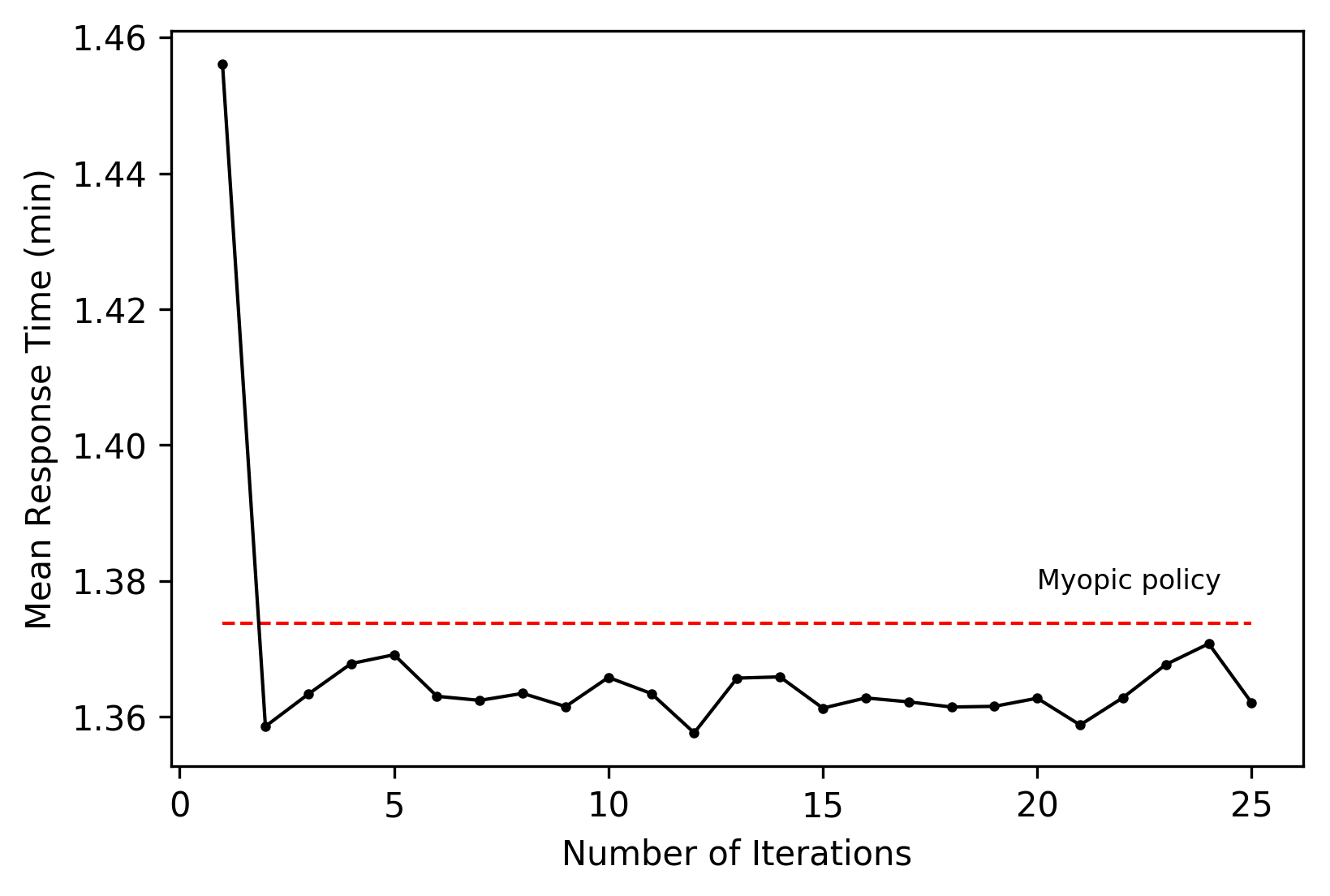}
\caption{Mean response time comparison for $N=5$ units.}
\label{fig_N_5}
\end{center}
\vspace{-0.4cm}
\end{figure}

\begin{figure}[htbp]
\begin{center}
\includegraphics[width=7cm]{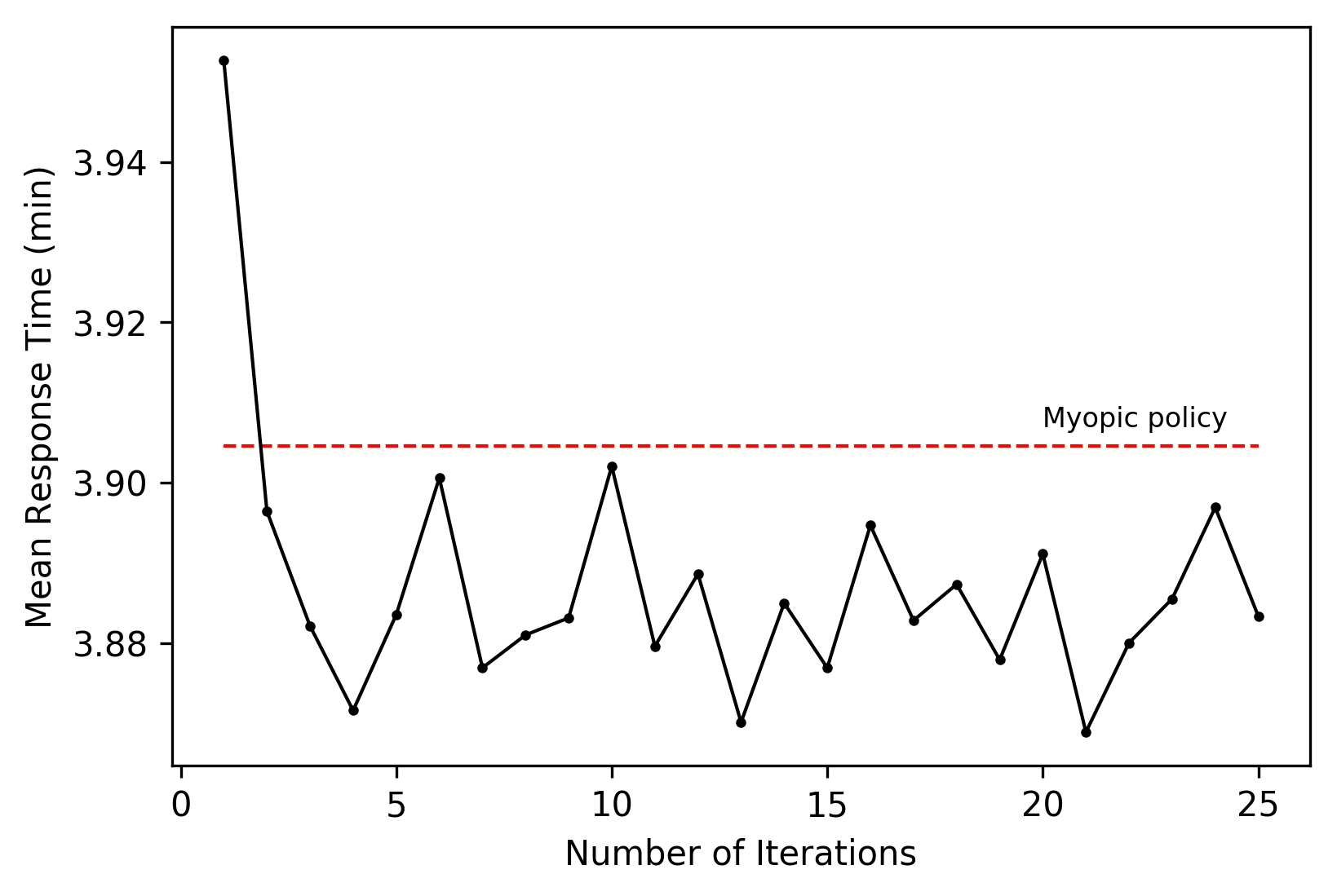}
\caption{Mean response time comparison for $N=10$ units.}
\label{fig_N_10}
\end{center}
\vspace{-0.4cm}
\end{figure}

\begin{figure}[htbp]
\begin{center}
\includegraphics[width=7cm]{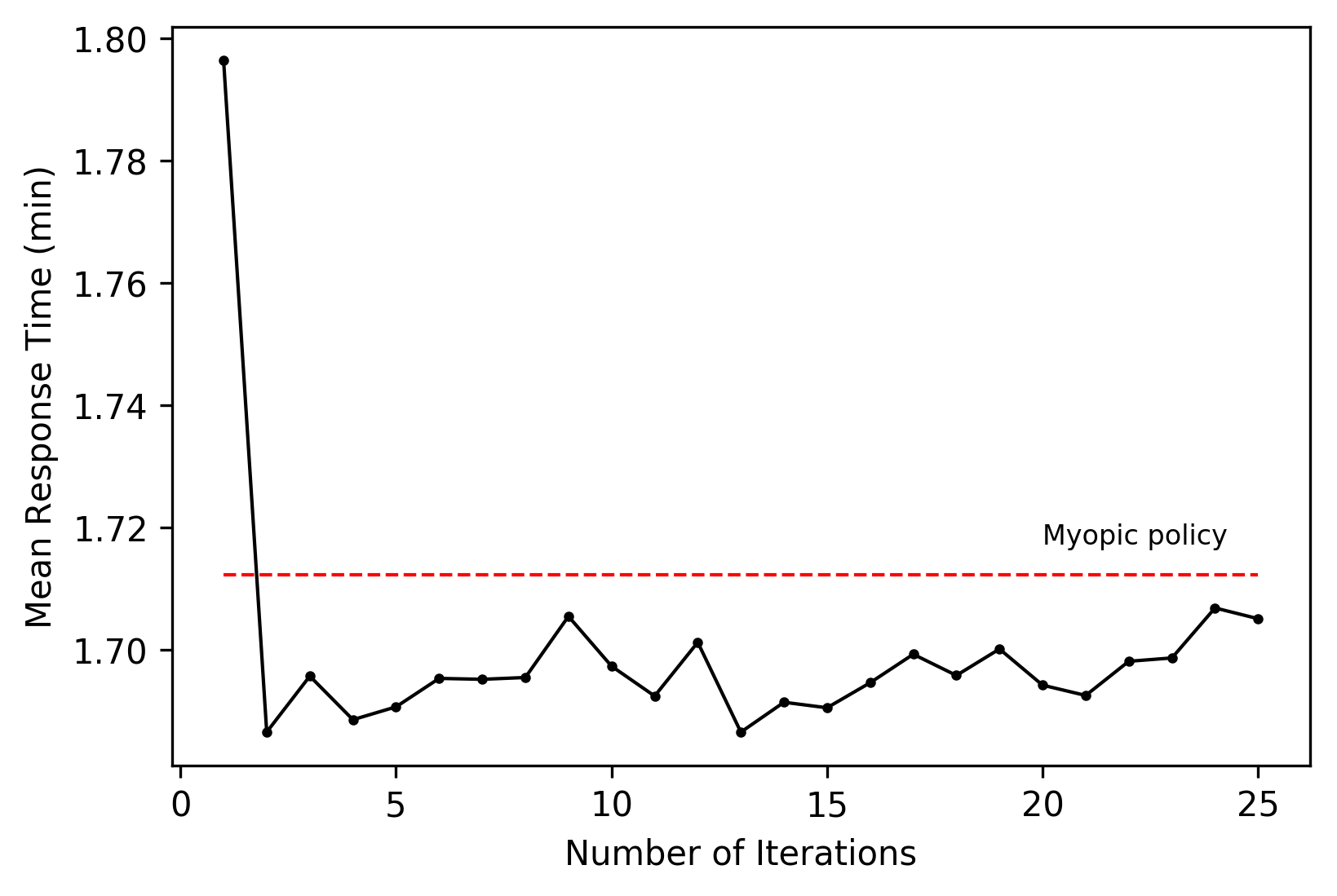}
\caption{Mean response time comparison for $N=15$ units.}
\label{fig_N_15}
\end{center}
\vspace{-0.4cm}
\end{figure}

\begin{figure}[htbp]
\begin{center}
\includegraphics[width=8cm]{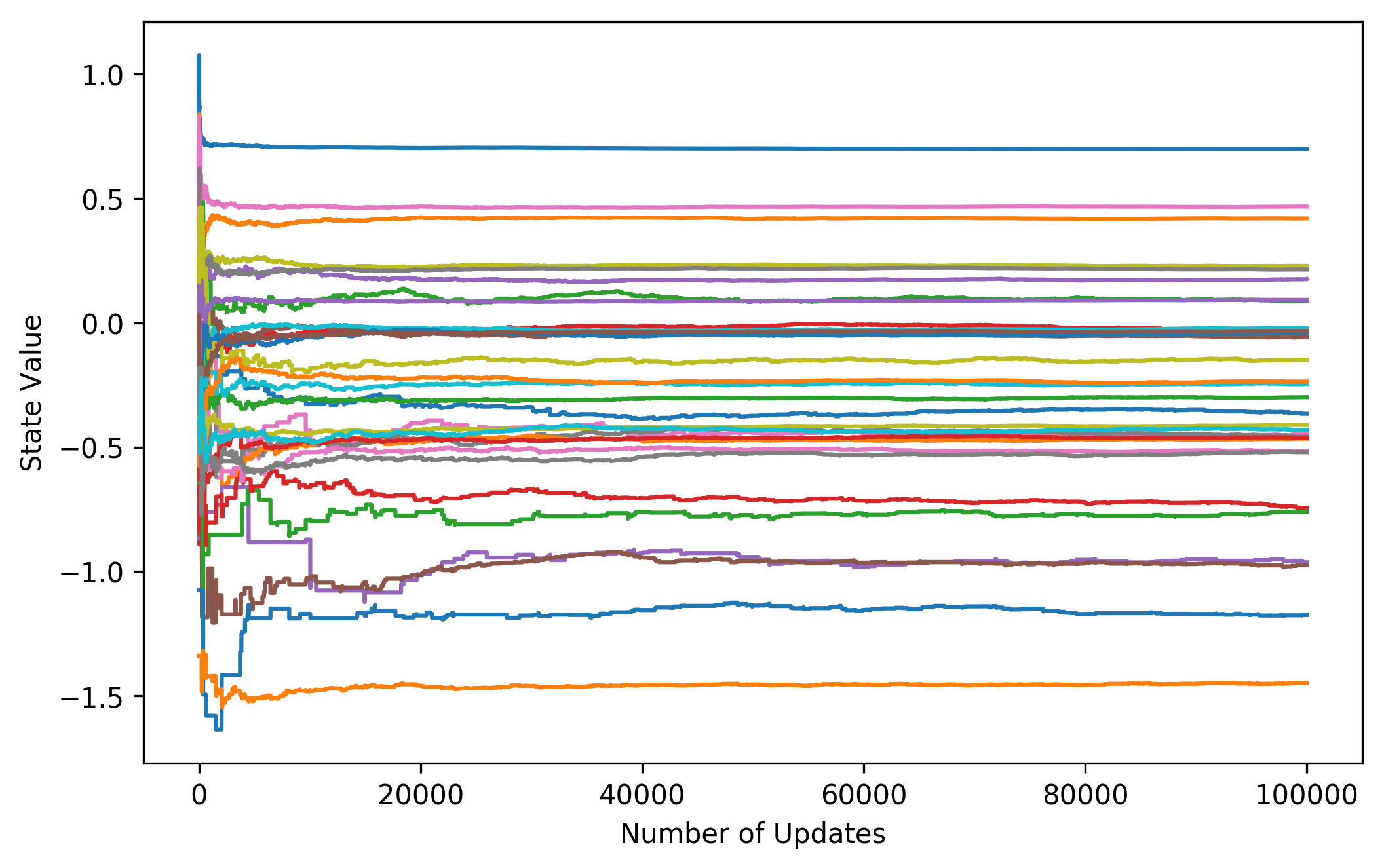}
\caption{State value updates in one TD-Learning iteration.}
\label{fig_SV_5}
\end{center}
\vspace{-0.4cm}
\end{figure}

From the three experiments we observe that the TD-Learning algorithm converges quickly in all cases as expected. We show the updates of the post-decision state values $\tilde{J}$ in one TD-Learning iteration for the case of $N=5$ in Figure \ref{fig_SV_5}. The resulted policies in all three cases outperform the myopic policy that always dispatch the closest available units. We also observe that our algorithm obtains a superior policy fairly quickly in about three iterations. In the case where $N=15$, solving the Bellman's equation requires solving a system with $31\times 2^{15}=1,015,808$ states, which is very computational costly, if not infeasible, by most modern computers. In contrast, our method obtains a good policy in less than two minutes in this case, and it applies virtually to systems with any sizes as guaranteed by its theoretical properties.  

Note that the policies obtained from our algorithm result in an average of three seconds reduction in terms of response time with no additional resources.  This suggests that myopic policies are unnecessarily wasting critical resources and slightly more intelligent policies can provide improved performance with little to no cost.

\section{Conclusion}
\label{sec_conclusion}

In this paper, we model the ambulance dispatch problem as an average-cost Markov decision process and aim to find the optimal dispatch policy that minimizes the mean response time. The regular MDP formulation has a state space of $(J+1)\cdot 2^N$, where $J$ is the number of demand nodes that partition the entire region and $N$ is the total number of ambulances in the system. The policy iteration is able to find the optimal policy. We propose an alternative MDP formulation that uses the post-decision states and reduce the state space to $2^N$. We show that this formulation is mathematically equivalent to the original MDP formulation. 

Even though the state space is reduced in the new formulation, due to the curse of dimensionality, its application is only restricted to small problems. We next present a TD-Learning algorithm based on the post-decision states that is guaranteed to converge to the optimal solution. In our numerically experiments, we show that the TD-Learning algorithm with post-decision states converge quickly and the policies obtained from our method outperforms the myopic policy that always dispatch the closest available unit in all cases. 

\bibliographystyle{named}
\bibliography{reference}
\end{document}